\documentclass[reqno]{amsart}
\usepackage{amssymb}
\newtheorem{thm}{Theorem}
\newtheorem{prop}[thm]{Proposition}

\newtheorem{lem}[thm]{Lemma}
\newtheorem{defn}[thm]{Definition}

\newcommand{\pair}[2]{{\langle#1,#2\rangle}}
\newcommand{\norm}[1]{\|#1\|}
\newcommand{\abs}[1]{\left| {#1} \right|}
\newcommand{\Restr}[2]{{#1}_{{|#2}}}

\begin{document}
\begin{abstract}
We prove the completeness of the Bethe ansatz eigenfunctions
of the Laplacian on a Weyl alcove with repulsive  boundary condition at the walls. For the root system of type {\bf A} this amounts to the result of Dorlas of the completeness of the Bethe ansatz eigenfunctions of the quantum Bose gas 
on the circle with repulsive delta-function interaction.
\end{abstract}
\title
[Completeness of the  Bethe ansatz on Weyl alcoves]
{Completeness of the  Bethe ansatz \\  on Weyl alcoves}
\author{E. Emsiz}
\date{\today}
\address{Instituto de Matem\'atica y F\'isica, 
Universidad de Talca, Casilla 747, Talca, Chile}
\email{eemsiz@inst-mat.utalca.cl}
\maketitle

\section{Introduction}
In a celebrated paper  Lieb and Liniger \cite{lieb-liniger} introduced the 
quantum Bose gas on the circle with repulsive delta-function interaction and 
using the Bethe ansatz method showed that it is  exactly solvable. C. N Yang and C. P. Yang  \cite{yang-yang} proved that the corresponding
Bethe ansatz equations are controlled by a strictly convex function. Fundamental progress 
was made by  Korepin \cite{korepin}, who  proved  Gaudin's \cite{gaudin, gaudin2} compact determinantal formula for
the $L^2$-norms of the Bethe ansatz eigenfunctions and Dorlas \cite{dorlas}, who 
proved the completeness and orthogonality of these functions 
using a lattice version of the quantum inverse scattering methods \cite{k-b-i}.

A crucial  insight by Gaudin \cite{gaudin,gaudin2} was that the
quantum  Bose gas  on the line with a delta-function 
potential has a natural generalization in the context of the root system of 
semi-simple complex Lie algebras. 
Gutkin and Sutherland \cite{gutkin-sutherland,gutkin}  
pushed these generalization further by introducing
affine root system version. The quantum Hamiltonian now has a potential
expressible as a weighted sum of delta-functions at the associated affine root hyperplanes.  In \cite{eos,eos2} the affine Weyl group invariant
quantum eigenvalue problem for these systems in the repulsive regime 
was studied using representation theory of the trigonometric Cherednik algebra and Dunkl operators.
From this perspective the quantum Bose gas of quantum particles  with repulsive delta-potential on the circle  essentially corresponds to the affine root 
system of type {\bf A}. A large part of  \cite{eos,eos2} deals with extending many of the  above properties of the quantum Bose gas to the root system versions (\cite{eos2} 
deals with quantum spin-particle systems).

In this paper we continue with the study of the general root system versions.  The main result is the completeness of the Bethe ansatz eigenfunctions in the sense that they span a dense subspace of the Hilbert space $L^2(A)$ of quadratically integrable functions with respect to the standard Euclidean measure on a Weyl alcove $A$.  
The  question of orthogonality and norm formulae for the root system versions turns out to be much harder to crack and   for general root systems  we only have 
conjectural statements.

The paper is organized as follows. In the following sections we introduce the quantum integrable systems and summarize 
the main results of \cite{eos,eos2} on the quantum spectral eigenvalue problem. 
The completeness will be deduced from a continuity argument of Dorlas \cite{dorlas} at zero coupling, which we recall in Section \ref{continuity}. An important part  is played by making sense (in Section \ref{mainproof}) of the formal quantum Hamiltonians as positive self-adjoint operators on $L^2(A)$. Section \ref{sobolev} is preparatory and deals with 
the non-interacting case  and the necessary Sobolev space theory to define the forms to make sense of the aforementioned formal Hamiltonians as positive self-adjoint operators on $L^2(A)$.

\section{Completeness}
\label{completeness}

We will freely use concepts and facts on root systems, see e.g. \cite{humphreys2} for a detailed exposition. Let $V$ be an Euclidean space of dimension $n$ and $R$ a finite, irreducible crystallographic root system spanning  $V$.  We denote $\langle \cdot,\cdot\rangle$ for the inner product on $V$ and $\|\cdot\|$ for the corresponding norm. The co-root of a root $\alpha\in R$ is $2\alpha/\norm\alpha^2$.
We fix a positive system $R^+$ and denote the corresponding basis of simple roots by $I=\{\alpha_1,\ldots,\alpha_n\}$.
We let  $\alpha_0=-\varphi$, where $\varphi$ denotes the highest root of $R^+$.
The corresponding Weyl alcove
\begin{equation}  \label{e:alcove}
  A=\{v\in V| 0< \langle \alpha,v\rangle < 1,\, \forall \alpha\in R^+\}
\end{equation}
is bounded by the walls
\begin{align}\label{e:walls}
V_0&=\{v\in V \; | \; \langle\alpha_0,v\rangle+1=0\},\\
V_i&=\{v\in V \; | \; \langle\alpha_i,v\rangle=0\},\quad i=1,\dots,n
\end{align}
The space of repulsive coupling constants $\mathcal K$ is the set  of all $k_\alpha\in (0,\infty)$, $\alpha\in R$ such that $k_\alpha=k_\beta$ whenever $\norm\alpha=\norm\beta$. We identify  $\mathcal K$ naturally with  $(0,\infty)$ (if $R$ has one root lenght) or  $(0,\infty)^2$ (respectively two root lenghts).

For any $k\in \mathcal K$ we will consider the following spectral problem for the Laplacian in the alcove $A$ with repulsive  boundary condition at the walls, i.e.
\begin{equation}
  \label{e:bvp1}
-\Delta f=E f   \quad \text{on $A$} ,
\end{equation}
  \begin{align}
\label{e:bvp2}
(\partial_{\alpha_i^\vee} f)(v)&=k_{\alpha_i} f(v),\,\quad v\in V_i\cap\overline A, 
\quad i=0,1,\dots,n,
  \end{align}
(note that $\alpha_i^\vee$ points inward to $A$)
where $\Delta$  refers to the Laplacian and $\partial_v$  to the partial derivative in the direction of $v$.

Let $s_i:V\to V$, $i=0,1,\dots,n$ denote the orthogonal reflection in the wall $V_i$:
\begin{align*}
s_0(v)&=v-(\langle \alpha_0,v\rangle+1) \alpha_0^\vee, \\
s_i(v)&=v-\langle \alpha_i,v\rangle \alpha_i^\vee, i=1,2\dots,n
\end{align*}
The reflections $s_1,\dots, s_n$ generate the  Weyl group $W$ associated with $R$ while  $s_0,s_1,\dots, s_n$ generate the affine Weyl group $W^a$. A second important presentation of $W^a$ is given by  $W^a=W\ltimes Q^\vee$,
with $Q^\vee$ the co-root lattice generated by $\alpha^\vee$, $\alpha\in R$, acting by translations on $V$.

It was shown in  \cite[Theorem 2.6]{eos},\cite{eos2} that the Bethe ansatz function 
\begin{equation}\label{e:BAF}
\phi^k_\lambda(v)=\frac{1}{\#W}\sum_{w\in W}c_k(w\lambda)e^{i\langle w\lambda,v\rangle},
\quad v\in\overline A,
\end{equation}
solves the boundary value problem \eqref{e:bvp1}, \eqref{e:bvp2} with eigenvalue $E=\norm\lambda^2$  if 
\begin{equation}
  \label{e:cfun}
c_k(\lambda)=\prod_{\alpha\in R^+}
\frac{\langle\lambda,\alpha^\vee\rangle-ik_\alpha}
     {\langle\lambda,\alpha^\vee\rangle}
\end{equation}
and the spectral parameter $\lambda\in V$ satisfies the Bethe ansatz equations
(BAE)
\begin{equation}\label{e:BAE}
e^{i\langle \lambda,\alpha_j^\vee\rangle} =\prod_{\beta\in R^+}
\left(
\frac{\langle \lambda,\beta^\vee \rangle+ik_\beta}
{\langle \lambda,\beta^\vee \rangle-ik_\beta}
\right)^{\langle \alpha_j^\vee, \beta \rangle},\quad j=1,2,\dots,n.
\end{equation}
 The  cone of dominant weights is denoted by
\[
\mathcal P^+=\{\lambda\in V \,\,\, | \,\,\,
\langle \lambda,\alpha^\vee\rangle \in\mathbb{Z}_{\ge 0},\quad \forall\alpha\in R^+\}.
\]
while $\mathcal P^{++}=\rho+\mathcal P^+$ is the cone of strictly dominant weights ($\rho$ is the half sum of positive roots).  The function 
$S_k:\mathcal P^{++}\times V\to \mathbb{R}$  defined by
\begin{equation}\label{master}
 S_k(\mu,v) = 
 \frac{1}{2}\norm{v}^2 - 2\pi\pair{v}{\mu}
   + \frac12\sum_{\alpha\in R}
    \norm{\alpha}^2 \int_0^{\pair{v}{\alpha^\vee}}\arctan\left(\frac{t}{k_\alpha}\right)dt
\end{equation}
is strictly convex and assumes a global minimum at  $\widehat\mu_k\in V$. Moreover it is known  \cite[Propositions 2.9 and 2.10]{eos} that  $\widehat \mu_k$  solves the BAE \eqref{e:BAE} and lies in the fundamental Weyl
chamber
\begin{equation}\label{chamber}
V_+=\{v\in V \, | \, \langle v,\alpha^\vee\rangle>0\quad
\forall\alpha\in R^+ \}.
\end{equation}

The main results of this paper states that the Bethe anstaz eigenfunctions are complete in the Hilbert space $L^2(A)=L^2(A,dv)$ of quadratically integrable functions on $A$ with respect to the Euclidean measure  $dv$:
\begin{thm}\label{complete}
The Bethe ansatz eigenfunctions 
 $\phi^k_{i\widehat\mu_k}$, 
$\mu\in \mathcal P^{++}$  span a dense subspace of $L^2(A)$.
\end{thm}

One sees easily that $\phi^k_{\widehat\mu_k}$ and $\phi^k_{\widehat\eta_k}$ are
orthogonal if $\norm{\widehat \mu_k}\neq \norm{\widehat\eta_k}$ (see e.g. \eqref{e:stokes}). Full orthogonality  is to the knowledge of the author only conjectural for general $R$,
as well as the following $L^2$-norm formula conjecture:
\begin{equation}\label{norm}
\frac{1}{\abs{A}}
\int_A\abs{\phi^k_{\widehat\mu_k}(v)}^2dv=
 \frac{\abs{ c_k({\widehat\mu_k})}^2 
   \det B^k_{\widehat\mu_k}}{\# W}, \quad \forall \mu\in \mathcal P^{++},
\end{equation}
Here  $B_v^k: V\times V\to \mathbb{R}$ is the Hessian
of $S_k(\mu,\cdot)$ at $v\in V$ (see also \cite[(9.3)]{eos}) and   $\abs{A}=\int_A dv$.

An interesting consequence of  \eqref{norm} would be 
$\lim_{k\to 0}   \det B^k_{\widehat\rho_k} =\#W$, 
an identity very similar to the limit formula~\cite[(3.5.14)]{he-sc} for the Heckman-Opdam hypergeometric functions associated to $R$.

As alluded to in the introduction, for $R$ of type {\bf A} the norm formulae and orthogonality conjectures were proved by Korepin and Dorlas, respectively. A new proof  of the orthogonality in this case was recently obtained  by van Diejen \cite{diejen06}  using an integrable lattice  dicretization of the  quantum Bose-gas on the circle. The norm formulae \eqref{norm} were also very recently checked \cite{BDM} for all $R$ with  rank less or equal to $3$, thus including  the important test case of $R$ of type  $G_2$.

The upgrade from type {\bf A} to other classical root systems amounts to adding particular reflection terms to the physical model,
see e.g.  \cite{gaudin,gaudin2, CrYo06} (and \cite{eos,eos2} for more references). 
For the exceptional root systems there is not such a convincing  
interpretation as a physical model (however, see  \cite{CrYo06}
for the exceptional root system $R$ of type $G_2$).

\section{Continuity}
\label{continuity}
Following Dorlas \cite{dorlas} we will deduce the completeness  from a continuity argument at zero coupling:
\begin{thm}[\cite{dorlas}]\label{dorlas}
Let $\{H_t\}_{t\in[0,\infty)}$
be a family of
positive self-adjoint operators acting on a Hilbert space $\mathfrak H$.
Assume that $t\mapsto H_t$ is monotonically nondecreasing
and right-continuous in the strong-resolvent sense and 
that $H_0$ has compact resolvent. 
Suppose that there exists a set of linearly independent eigenfunctions
$\{\phi^t_n\}_{n\in N}$ for $H_t$, parameterized by  a discrete set $N$,  
which depends continuously on $t\in [0,\infty)$,
and which is complete at $t=0$. Then 
the linear subspace
spanned by the $\phi^t_n$ ($n\in N$) is  dense in $\mathfrak H$
for all $t\in[0,\infty)$.
\end{thm}
In \cite{eos,eos2} the eigenvalue value problem \eqref{e:bvp1}, \eqref{e:bvp2} 
was interpreted as the $W^a$-invariant spectral problem of the following formal quantum Hamiltonian
\begin{equation}\label{e:ham}
\mathcal H_k=-\Delta + \sum_{\substack{\alpha\in R, m\in\mathbb{Z}}}k_\alpha \delta(\pair{\alpha}{\cdot}+m)
\end{equation}
and was studied with the aid of trigonometric Cherednik algebra and Dunkl operators.
Here  $W^a$ acts on the space of $\mathbb{C}$-valued functions $f$ as follows,
 \begin{equation}  \label{e:usual}
(wf)(v)=f(w^{-1}v), \quad (w\in W^a, v\in V).
\end{equation}
We will show that for the $W^a$-invariant theory the formal $\mathcal H_k$ can be interpreted as self-adjoint operators on the fundamental domain $A$ for the action of  $W^a$ on $V$,
and moreover they will be the operators $H_k$ in Theorem \ref{dorlas} (see Section \ref{mainproof} for the exact statement). The functions $\phi^t_n$, $n\in N$  will be $\phi^k_{\widehat\mu_k}$, $\mu \in \mathcal P^{++}$. As for the case $k\equiv 0$ (free boundary value problem), note that the plane wave
\begin{equation}\label{phi0}
\phi_\lambda^0(v)=\frac{1}{\#W}\sum_{w\in W}e^{iw\lambda(v)},\quad v\in \overline A,
\end{equation}
solves \eqref{e:bvp1}, \eqref{e:bvp2} iff   $\lambda=\widehat\mu_0=2\pi(\mu-\rho)$, $\mu\in\mathcal P^{++}$. The continuity of the eigenfunctions, one of the conditions of Theorem \ref{dorlas} is guaranteed by:
\begin{prop}\label{phicont}
Let $\mu\in \mathcal P^{++}$.  Then $k\mapsto \phi^k_{\widehat\mu_k}$
defines a continuous map from $\mathcal K\cup\{0\}$ to $L^2(A)$.
\end{prop}
Since $(v,k)\mapsto S_k(\mu,v)$ is differentiable in $k\in \mathcal K$, $v\in V$
the implicit function theorem shows that
$k\mapsto \hat{\mu}_{k}$ is smooth from $\mathcal K$ to $V_+$, and therefore
$k\mapsto \phi^{k}_{\widehat\mu_k}$ is continuous from $\mathcal K$ to $L^2(A)$.
The problematic point is continuity in the boundary point $k\equiv 0$.
We will actually not use any differentiability, Proposition \ref{phicont} will be proved by  a purely topological argument, based on the following lemma (which we formulate in greater generality than strictly necessary  since the proof is not more difficult).
\begin{lem}\label{gammacont} 
Let $X$ and $Y$ be locally compact Hausdorff spaces with $Y$ first countable.
Let $f$ be a continuous real-valued function on $X\times Y$
such that $f(\cdot,y)$ has a 
unique global minimum at $\gamma_y$
for  $y\in Y$. Define the map
$\gamma:Y\to X$  by $y\mapsto \gamma_y$. 
Assume also there is a compact $D\subset X$ such that  
$\gamma(Y)\subset D$. Then $\gamma$ is continuous. 
\end{lem}
\begin{proof}
Assume  $\gamma$ is not continuous in a point $\eta\in Y$. We
will show that this leads to a contradiction. 
There is a neighborhood $N$ of $\gamma_\eta$ with compact closure and a sequence $y_1,y_2,\dots$ 
in $Y$ with $y_j\to \eta$ and $\gamma_{y_j}\in X\backslash N$ for all $j$,
and where $X\backslash N$ is the complement of $N$ in $X$.
Defining  $m_0=\min_{x\in D \cap (X\backslash N)} (f(x,\eta)-f(\gamma_{\eta},\eta))$,
we have $m_0>0$ by definition of $\gamma_{\eta}$ and compactness of $D$.
Choose an $\varepsilon$ such that $0<\varepsilon<m_0/4$. 
Choose a neighborhood $U$ of $\eta$  with compact closure such that 
 $\abs{f(x,\eta) - f(x,y)}\le \varepsilon$
for $y\in U$ and $x\in D$.
Hence for $y\in U$ 
we have
$  f(\gamma_y,\eta) - f(\gamma_y,y)\le\varepsilon.$
Furthermore by the definition of $\gamma_y$ we have
$f(\gamma_y,y)  \le  f(\gamma_{\eta},y)$,
and  continuity of $f$ in $y$ implies that there is a 
neighborhood $U'$ of $\eta$ such that 
$f(\gamma_{\eta},y) \le f(\gamma_{\eta},\eta) + \varepsilon$ for $y\in U'$,
thus
$f(\gamma_y,y)  \le f(\gamma_{\eta},\eta) + \varepsilon$ for $y\in U'$.
Therefore 
$
 f(\gamma_y,\eta) - f(\gamma_{\eta},\eta) \le 2\varepsilon <m_0/2
$
for $y\in U\cap U'$,
which leads to a contradiction  for  $y=y_j$ and $j$ large
by the definition of $m_0$.
\end{proof}

\begin{proof}[Proof of Proposition \ref{phicont}]
Let $X=V$, $Y=\mathcal K\cup\{0\}$ and $f(v,k)=S_k(\mu,v)$, where
$ S_0(\mu,v) = \frac{1}{2}\norm{v}^2 - 2\pi\pair{v}{\mu}
   + \pi \sum_{\alpha\in R^+} \abs{\pair{v}{\alpha}}$.
Then $S_0(\mu,v)$ is continuous (but not differentiable),
strictly convex in $v$ and has a unique global minimum that is obtained at
$\widehat\mu_0=\mu-\rho$. Lemma \ref{gammacont} with \cite[Propositions 2.10]{eos} yields
that  $k\mapsto \widehat{\mu}_{k}$ is continuous on $\mathcal K\cup\{0\}$, and  in particular 
$\lim_{k \to 0} \widehat{\mu}_{k} = 
\widehat{\mu}_0=2\pi(\mu-\rho)$. By  Section 9 of \cite{eos} $\widehat{\mu}_k\in V$ is determined by the equation
$\widehat{\mu}_k+\sigma_{\widehat{\mu}_k}^k=2\pi\mu$
where $\sigma_\lambda^k=2\sum_{\alpha\in R^+}
\arctan\left(\pair{\lambda}{\alpha^\vee}/k_\alpha\right)\alpha$.
Taking the inproduct of this  with a  $v\in V_+$ gives
 $\arctan(\pair{\widehat\mu_k}{\alpha^\vee}/k_\alpha) \to \pi/2$ as $k\to 0$, and whence $\pair{\widehat\mu_k}{\alpha^\vee}/k_\alpha\to+\infty$ as $k\to 0$. In particular $c_k(\widehat\mu_k)$ is continuos in $k\in \mathcal K\cup \{0\}$, which together with \eqref{e:BAF} yields the proposition.
\end{proof}

In the following section we will make sense of $\mathcal H_0=-\Delta$ as a self-adjoint positive operator  $H_0$. The form to make sense  of \eqref{e:ham} as a self-adjoint operator on $L^2(A)$ will be defined in the final section. It will be a  perturbation  of the form corresponding to  $H_0$. For this we will show  in
the following section that the form domain  of the form corresponding to  $H_0$
is  the Sobolev space of once  weak differentiable functions on $A$ while the form itself is equal to the Sobolev  inner product.
\section{The Laplacian  on the Weyl alcove}
\label{sobolev}
The domain of any operator $H$ will be denoted by $D(H)$ and its form and form 
domain by $q_H$ and $Q(H)$, respectively.

One checks easily that a function $f\in D_0=\{\Restr{f}{\overline A}|f\in C^\infty(V)^{W^a}\}$ satisfies
\begin{equation}\label{e:freejump}
(\partial_{\alpha_i^\vee}f)(v)=0, \quad  v\in V_i\cap\overline A,\;\; i=0,\dots,n
\end{equation}
By Stokes' theorem we have for suitable $f,g$,
\begin{equation}
  \label{e:stokes}
\int_A ((\Delta f)g   -f\Delta g)=
\sum_{j=0}^n\frac1{\norm{\alpha_j^\vee}}\int_{V_j\cap \overline A}
( f\overline{\partial_{\alpha_j^\vee}g}-(\partial_{\alpha_j^\vee}f) \overline g).
\end{equation}
Whence $-\Delta$ with domain $D_0$ is symmetric on $L^2(A)$, and moreover also positive by another application of Stokes' theorem. Furthermore: 
\begin{prop}\label{freeHam} ($-\Delta,D_0)$ is essentially 
self-adjoint and $D_0$ is an operator core for the unique positive self-adjoint extension $(H_0,D(H_0))$. 
 Moreover the $\phi^0_{\mu}\in D_0\subset D(H_0)$, $\mu\in 2\pi \mathcal P^+$ form a  complete set of orthogonal functions and
$ H_0 \phi^0_{\mu} = \norm{\mu}^2 \phi^0_{\mu}$.  Furthermore $H_0$ has compact resolvent.
\end{prop}

\begin{proof}
The first statement follows from the second one and  \cite[Theorem X.39]{reed-simon-II}.  We have $\phi^0_\mu\in D_0$  because \eqref{phi0} extends directly to a $W^a=W\ltimes Q^\vee$-invariant function in $C(V)^{W^a}$.
 Orthogonality of the $\phi^0_\mu$ is shown by an explicit
calculation and $ H_0 \phi^0_{\mu} =-\Delta\phi^0_\mu= \norm{\mu}^2 \phi^0_{\mu}$. Density of the $\phi^0_\mu$  follows from an elementary application of the Stone-Weierstrass theorem. The last statement follows from \cite[Theorem XIII.64]{reed-simon-IV}. 
\end{proof}

For a $f\in L^2(A)$ and  $\eta\in V$ one says that $\partial_\eta f\in L^2(A)$ weakly  provided that there is a
$g\in L^2(A)$  such that $\partial_\eta f=g$ as distributions on $A$.
It is a standard result in Sobolev space theory that 
$W^{1,2}=
\{f\in  L^2(A)\; | \; 
\partial_\eta f\in L^2(A),\;  \forall \eta\in V\} $
is a Hilbert space with inner product
\begin{equation}\label{sobolevnorm}
(f,g)_{W^{1,2}} =
\int_{A} f(v)\overline{g(v)} dv +
 \sum_{j=1}^n \int_{A} (\partial_{\eta_j} f)(v) 
\overline{(\partial_{\eta_j}g)(v)} dv,
\end{equation}
where $\{\eta_1,\dots,\eta_n\}$ denotes an orthonormal basis for $V$ 
(\eqref{sobolevnorm} is in particular independent of the choice of
orthonormal basis).  Stokes' theorem with \eqref{e:freejump} 
yields that  $h_0(f,g) = (f,g)_{W^{1,2}}$ for $f,g\in D_0\subset Q_0$.
Making  use of the geometry of the affine hyperplane arrangement associated with the  root system $R$ we can prove for the symmetric domain $A$:
\begin{prop}\label{Q0W}
The form domain $Q_0$ of $h_0$ is  equal to the Sobolev 
space  $W^{1,2}$ of once weak differentiable functions on $A$
and $h_0=(\cdot,\cdot)_{W^{1,2}}$.
\end{prop}
\begin{proof}
Since $D_0$ is an operator core for $(H_0,D(H_0))$
 (Proposition \ref{freeHam}) and $h_0=(\cdot,\cdot)_{W^{1,2}}$ on $D_0$ 
it suffices to show that $D_0$ is a dense subspace of $W^{1,2}$
in the Sobolev topology for the first statement. 
The second statement then follows because $D_0$ is a form core for $(h_0,Q_0)$.
Since 
 $C^\infty(\overline{A})= \bigr\{\Restr{f}{\overline{A}} 
  \; | \; f\in C^\infty(V) \bigr\}$
is dense  in $W^{1,2}$ (see e.g. \cite[Theorem~I.3.6]{wloka}),
  it  suffices to show 
that every $\phi\in C^\infty(\overline{A})$ can be 
approximated in $W^{1,2}$ by functions from $D_0$.
Fix therefore a $\phi\in C^\infty(\overline{A})$.
Denote the unique $W^a$-invariant extension of $\phi$ to $V$ by  $u$, thus $u\in C(V)^{W^a}$ (see \eqref{e:usual}).

We will use a regularization procedure. For a $r>0$ and $v\in V$ let $\bar B(v,r)=\{v'\in V|\norm{v'-v}\le r \}$.
Choose a $j\in C^\infty_c(V)^{W}$ with $j\ge 0$ and support
in $\bar B(0,1)$, and $\int_V j(v) dv = 1.$
For $\varepsilon>0$ define  $j_\varepsilon \in C^\infty_c(V)^{W}$ by
$j_\varepsilon(v)=j(v/\varepsilon)/\varepsilon^n$.  Consider the regularization 
$u_\varepsilon \in C^\infty(V)$ of $u$,
$
u_\varepsilon (v) = \int_{V} j_\varepsilon(v-v') u(v') dv'
 =  \int_{V} j_\varepsilon(v') u(v-v') dv'.
$
Since $u\in L^1_{loc}(V,dv)^{W^a}$ we have $u_\varepsilon\in C^\infty(V)^{W^a}$ and
in particular $\Restr{u_\varepsilon}{\overline{A}}\in D_0$.  We shall show that  $\Restr{u_\varepsilon}{\overline{A}}\to \phi$ in the Sobolev norm as $\varepsilon\searrow 0$.

For any $\eta\in V$ consider the function in $u^\eta\in L^1_{loc}(V,dv)$ defined uniquely  by $u^\eta=\partial_\eta u$ on $V_{\text{reg}}$. We claim that 
$\partial_\eta u_\varepsilon= (u^\eta)_\varepsilon$ for all $\varepsilon>0$.
For  $t>0$  consider the difference-quotient function
$g_{t,\eta}(v) = (u(v+t\eta) - u(v))/t$ for  $v\in V$.
Then $g_{t,\eta}$ is continuous on $V$ and $Q^\vee$-periodic.
It is easily seen that for every $v\in V$ and $0<t < \text{diam}(A)$
the line segment between $v$ and $v+t\eta$ lies in at most $\#W/2$
alcoves.
Since  $\Restr{u}{\overline C}\in C^\infty(\overline{C})$ for all alcoves $C=wA$ ($w\in W^a$),
multiple applications of the mean value theorem (but at most $\#W/2$ times)  together with the $W^a$-invariance 
and continuity of $u$ yields the bound $\#W M_\eta/2$ for
 $\abs{g_{t,\eta}}$ on $V$
for $t\in (0,\text{diam}(A))$ and
where $M_\eta$ denotes the supremum of $\partial_\eta u$ on $A$.
Whence  $\abs{j_\varepsilon g_{t,\eta}(v-\cdot)}$ is bounded by
 $(\#W M_\eta)/2\sup_{V}j/(\varepsilon^n)$ (and has support in $\bar B(0,\varepsilon)$) for all $v\in V$.  Since
 $\lim_{t\to 0} g_{t,\eta} = u^\eta$ pointwise on $V_{\text{reg}}$,
Lebesgue's dominated convergence theorem  (integrating over $\bar B(0,\varepsilon)$)
yields then indeed
 $\partial_\eta u_\varepsilon = (u^\eta)_\varepsilon$ for all $\varepsilon>0$.

Together with  $u_\varepsilon\to u$  uniformly on compact subsets of $A$, which
follows from  $\Restr{u}{A}\in C(A)$ and an elementary uniform continuity argument, it follows that $\partial_\eta u_\varepsilon\to \partial_\eta u$ point-wise on $A$.
Note that if a locally integrable function  is bounded by a $M\ge 0$, 
then so does its regularization.
Applying this observation to $u^\eta$  and using $\partial_\eta u_\varepsilon = (u^\eta)_\varepsilon$ shows that $\abs{\partial_\eta(u_\varepsilon - u)}$ is bounded by
$2M_\eta$
uniformly  on $A$.  A simple application of 
Lebesgue's dominated convergence theorem now gives
  $\norm{\partial_\eta u_\varepsilon-\partial_\eta u}_{L^2(A)}\to 0$ as $\varepsilon\to 0$.
Furthermore  $\norm{u_\varepsilon-u}_{L^2(A)}\to 0$ as $\varepsilon\to0$ since  $u_\varepsilon \to u$ uniformly on $\overline{A}$ and whence  $\Restr{u_\varepsilon}{A}\to \Restr{u}{A}=\phi$
in the Sobolev norm,  concluding the proof.
\end{proof}

In the following final section we will finish the  proof of Theorem \ref{complete} on the completeness of the Bethe ansatz eigenfunctions on the Weyl alcove $A$.

\section{Proof of the completeness}
\label{mainproof}
We denote the restriction of the Euclidean measure on $V$ to $\partial{A}$   
also by $dv$. The following non-trivial result is a special case  of \cite[Theorem I.8.7]{wloka},
\begin{thm}[Trace operator]\label{trace} 
There exist a continuous linear operator 
\begin{equation}
\mathcal B: W^{1,2}\to L^2(\partial A,dv)
\end{equation}
such that  $\mathcal B\phi = \Restr{\phi}{\partial A}$ for
 $\phi\in W^{1,2}\cap C^1(\overline{A})$, where 
$C^1(\overline A)=\{\Restr{f}{\overline A}|f\in C^1(V)\}$.
\end{thm}

We now are able to  define the form to make sense of \eqref{e:ham} as a positive self-adjoint operator.
\begin{defn}\label{hk}
We denote by  $\delta_k$ the following form with domain $W^{1,2}$,
$$ \delta_k(f,g)= 
\sum_{j=0}^n \frac{k_{a_j}}{\norm{a_j^\vee}}
\int_{ \overline{A}\cap V_j} \mathcal Bf(v) \overline{\mathcal Bg(v)}\,\,dv
$$
and by $h_k$ the form $h_0+\delta_k$ with domain $W^{1,2}$ ($=Q_0$ 
by Proposition \ref{Q0W}).
\end{defn}
By continuity of  $\mathcal B$ there is  a $c>0$ such that
\begin{equation}\label{e:sobineq}
  \delta_1(f,f)\le c \, ( (f,f)_{W^{1,2}} +(f,f)_{L^2(A)}),\quad \forall f\in W^{1,2}
\end{equation}

For root system of type {\bf A}  our $\delta_k$ differs from Dorlas' (see \cite[(2.7)]{dorlas}) in the sense that his $\delta_k$ has no integration over an affine wall, which is explained by the fact that we work 
in the center-of-mass coordinates.

We write $k\le k'$ if $k_\alpha\le k'_\alpha$ for all $\alpha\in R$.
Following Dorlas \cite{dorlas} we are now in position to make sense of the formal Hamiltonian~\eqref{e:ham} as 
a positive self-adjoint operator.
\begin{prop}\label{Hk} 
(i) There is a unique positive self-adjoint
operator $H_k$ on $L^2(A)$ with form
 $(h_k, W^{1,2})$. Furthermore $D_0$ is a form core for $h_k$.\\
(ii) The  $H_k$ are non-decreasing in $k$, i.e. $H_k\le H_{k'}$ if $k\le k'$.\\
(iii) $k\mapsto H_k$ 
is right-continuous in $k\in \mathcal K\cup\{0\}$ in the strong-resolvent sense.\\
(iv) Let $\mu\in \mathcal P^{++}$.
Then $\phi^k_{\widehat\mu_k}\in D(H_k)$ and
 $H_k \phi^k_{\widehat\mu_k}= \norm{\widehat\mu_k}^2 \phi^k_{\widehat\mu_k}$
\end{prop}

\begin{proof}
The uniqueness of $H_k$  follows because a self-adjoint operator is
uniquely defined by its form by the first representation theorem
 \cite[Theorem~ VI.2.1]{kato}.  
For the existence we use induction on $k_{\text{max}}=\max\{k_\alpha|\alpha\in R\}$, starting from $k\equiv 0$ using Proposition \ref{freeHam}.  Assume therefore that $H_{k^0}$ is defined for some $k^0\ge 0$ with the properties mentioned in (i). Let $\beta$ be the form $\delta_{k-k^0}$. Thus $h_k=h_{k^0}+\beta$ is a perturbation of $h_{k^0}$ by $\beta$. Then it follows immediately from the non-decreasing of the 
forms $h_k$ and \eqref{e:sobineq} that $\beta(f,f) \le c (k-k^0)_\text{max}(h_{k^0}(f,f)+(f,f))$ for all $f\in W^{1,2}$.
The KLMN theorem \cite[Theorem~X.17]{reed-simon-II},  \cite[Theorem~VI.3.9]{kato}  now yield the existence of the operators $H_k$ for  those $k\in \mathcal K$ satisfying  $ k^0 \le k < k^0+1/c$. Whence by  induction they are defined for all $k\in \mathcal K$, yielding (i).
Statement (ii) follows from the non-decreasing of the $h_k$. 
Statement (iii) follows from  \cite[Theorem VI.2.1(iii)]{kato} and
$h_k(\phi^k_{\widehat\mu_k},\psi) = (-\Delta\phi^k_{\widehat\mu_k},\psi)_{L^2(A)}$ for all $\psi\in D_0$, which itself follows from the definition of $h_k$, symmetry of $h_k$ and \eqref{e:stokes}. 
As for (iv), let $k_1\ge k_2 \ge \dots \ge k^0$ be a sequence in $\mathcal K\cup \{0\}$ s.t. $k_j\to k^0$.
Then $h_{k_1} \ge h_{k_2} \ge \dots \ge 0$ as forms on $W^{1,2}$.
Now  $ h_{k^0}(\phi,\phi) = \inf_{k>k^0}  h_{k}(\phi,\phi)$,  $\forall\phi\in D_0$, yielding $h_{k^0}=\lim_{j\to\infty} h_{k_j}=\inf_{j\in\mathbb{N}} h_{k_j}$ by the first part of 
\cite[Theorem S16]{reed-simon-I} and (iv)  follows  by the second part of 
\cite[Theorem S16]{reed-simon-I}.
\end{proof}

\begin{proof}[Proof of Theorem \ref{complete}]
 Denote the space of $\mathbb{R}$-valued polynomials on $V$ over $\mathbb{R}$ by $P(V)$. To any $p\in P(V)$ one can associate in a natural way a constant coefficient differential operator $\partial_p$. Since $\partial_p \phi^k_\lambda = p(\lambda)\phi^k_\lambda$  for $p\in P(V)^W$ (see \eqref{e:usual}) and $P(V)^W$ separates the orbits $V/W$ it follows that
the $\phi^k_{\widehat\mu_k}$, $\mu\in \mathcal P^{++}$ must be linearly independent. The 
theorem now follows from Theorem~\ref{dorlas} 
applied to $(H_{tk})_{t\in[0,\infty)}$, $(\phi^{tk}_{\widehat\mu_{tk}})_{\mu\in \mathcal P^{++}}$ and using Propositions \ref{phicont}, \ref{freeHam} and \ref{Hk}.
\end{proof}

\noindent
{\bf Acknowledgments.} 
The research reported in this paper 
was supported in part by the Fondo Nacional de Desarrollo Cient\'ifico y Tecnol\'ogico (FONDECYT)  grant \#3080006.
The author would like to thank  Eric Opdam  and Jasper Stokman for many useful discussions and helpful comments. 
He also thanks  Jan Felipe van Diejen for interesting discussions.


\providecommand{\bysame}{\leavevmode\hbox to3em{\hrulefill}\thinspace}
\providecommand{\MR}{\relax\ifhmode\unskip\space\fi MR }
\providecommand{\MRhref}[2]{\href{http://www.ams.org/mathscinet-getitem?mr=#1}{#2}}
\providecommand{\href}[2]{#2}

\end{document}